\documentclass[a4paper,12pt]{article}
\usepackage{amsfonts}
\usepackage{amsmath,amsthm,amsfonts,amssymb,bm} 
\usepackage{graphicx}                           
\usepackage[usenames]{color}
\usepackage{makeidx}
\usepackage[nottoc]{tocbibind}
\usepackage{calligra}
\usepackage{times}
\usepackage{fancyhdr}
\usepackage{calc}
\usepackage{pdflscape}
\usepackage[titletoc,toc,title]{appendix}

\newcommand*{\F}{\mathcal{F}}

\newcommand{\R}{\ensuremath{\mathbb{R}}}
\newcommand{\IP}{\ensuremath{\mathbb{P}}}

\newcommand*{\ud}{\mathrm{d}}
\makeatletter
\newcommand*{\rom}[1]{\expandafter\@slowromancap\romannumeral #1@}
\makeatother

\DeclareMathOperator*{\essinf}{ess\,inf}

\usepackage{tikz}
\tikzset{
    vertex/.style = {
        circle,
        fill            = black,
        outer sep = 2pt,
        inner sep = 1pt,
    }
}
 \tikzstyle{mybox} = [draw=blue!70, fill=blue!10, thin,
    rectangle, rounded corners, inner sep=10pt, inner ysep=20pt]
\tikzstyle{fancytitle} =[fill=green!20, text=blue!100]
\usetikzlibrary{arrows,decorations.pathmorphing,backgrounds,positioning,fit,petri}

\usepackage{eurosym}
\usepackage[left=2cm, right=2cm]{geometry}

\usepackage{natbib} 
 \bibpunct{(}{)}{,}{}{,}{,} 
\usepackage[affil-it]{authblk}
\usepackage[ pdftex, plainpages = false, pdfpagelabels,
                 bookmarks,
                 bookmarksopen = true,
                 bookmarksnumbered = true,
                 breaklinks = true,
                 linktocpage,
                 pagebackref,
                 colorlinks = true,
                 linkcolor = blue,
                 urlcolor  = blue,
                 citecolor = blue,
                 anchorcolor = green,
                 hyperindex = true,
                 pdfstartview=FitH,
                 hyperfigures,
                 implicit,
                ]{hyperref}

\usepackage{subfigure}
\usepackage{tikz}

\theoremstyle{plain} \newtheorem{thm}{Theorem}[section]
\theoremstyle{plain} \newtheorem{lem}[thm]{Lemma}
\theoremstyle{plain} \newtheorem{coll}[thm]{Corollary}
\theoremstyle{plain} \newtheorem{defi}{Definition}[section]
\theoremstyle{plain} 
\theoremstyle{plain} 
\theoremstyle{plain} 
\theoremstyle{plain} 
\theoremstyle{plain} \newtheorem{prop}{Proposition}[section]
\theoremstyle{plain} 
\theoremstyle{plain} \newtheorem{Problem}{Problem}[section]
\usepackage{eso-pic}

\newcommand\Nchapter[1]{%
  \if@mainmatter%
    \@mainmatterfalse%
    \chapter{#1}%
    \@mainmattertrue%
  \else
    \chapter*{#1}%
  \fi}

\makeindex

\renewcommand{\baselinestretch}{1.5}
\usepackage{booktabs}
\usepackage{enumerate}
\usepackage[section]{placeins}

\allowdisplaybreaks

\begin{document}
\title{Horizon-unbiased Investment with Ambiguity\footnote{
We are grateful for the funding    from the NSF of China (11501425 and 71801226).  }}
 \author{Qian Lin\thanks{ Email: linqian@whu.edu.cn}}
\affil{School of Economics and Management,   Wuhan University, China}
\author{Xianming Sun\thanks{ Email: Xianming.Sun@zuel.edu.cn}}
\affil{School of Finance, Zhongnan University of Economics and Law, China}
\author{Chao Zhou\thanks{ Email: matzc@nus.edu.sg}}
\affil{Department of Mathematics, National University of Singapore, Singapore }
\maketitle
\begin{abstract}
In the  presence of ambiguity on the driving force of  market randomness, we consider the   dynamic portfolio choice  without any predetermined  investment horizon.
The investment criteria is formulated as a robust forward  performance process, reflecting an investor's dynamic preference.
We show that   the market risk premium and the utility risk premium jointly    determine  the investors' trading direction and the  worst-case scenarios of the risky asset's mean return and   volatility.
The closed-form formulas  for the optimal investment strategies are given in the special settings of the CRRA preference.

\end{abstract}

\section{Introduction}

Dynamic portfolio choice problems usually envisage an investment setting where an investor  is exogenously assigned  an investment performance criteria and stochastic models for the price processes of risky assets.
However, the investor may extemporaneously change the investment horizon,  consistently update her preference   with the market evolution, and conservatively invest due to  ambiguity on the driving force of   market randomness or the dynamics     of the risky assets.
Motivated by these  investment realities,
we study a robust horizon-unbiased portfolio problem in a continuous-time framework.


In the seminal work of \cite{Merton1969}, continuous-time portfolio choice   is formulated as a stochastic control problem   to maximize the expected utility at a specific investment horizon by searching for the optimal strategy in an admissible strategy space.
Note that if the investor has two candidate investment horizon $T_1$, $T_2$, $(T_2>T_1>0)$, the resulting optimal strategies associated with these two horizons are generally not consistent over the common time interval $[0,T]$, $(T\le T_1<T_2)$  \citep{Musiela2007}.
Hence, Merton's framework is neither  suitable for the case where    an investor may extend or shorten   her initial investment horizon, nor  the case where the investor may update   her preference  in accordance to the accumulated market information.
In these quite realistic settings, the investor needs an optimal strategy which is independent of the investment horizon and reflects her dynamic preference in time and wealth.
The horizon-unbiased utility or  forward performance measure, independently proposed by \cite{Choulli2007,Henderson2007,Musiela2007}, provides a portfolio framework   satisfying the aforementioned requirements.
In such framework, an investor specifies   initial preferences (utility function), and then propagates them \emph{forward} as the financial market evolves.
This striking characteristic contrasts the  portfolio choice based on the forward performance measure from that in   Merton's framework,   in which
intertemporal  preference is   derived  from the terminal utility function in a \emph{backward} way.
\cite{Musiela2010a} specify the  generic forward performance measure  as a stochastic flow $U=U(t,x)_{t\ge0}$,  taking time $(t)$ and wealth $(x)$ as arguments.
The randomness of the forward performance measure is driven by the   Brownian motion which is the same as the driving force of the randomness of asset price.
It implies that the driving force  of market randomness is  simultaneously embedded into the investor's preference and the risky asset price process.
Such modeling approach   implicitly assumes that
the Brownian motion represents the essential source of risk behind the financial market and the risky assets.
 Especially,
the volatility of a forward performance measure reflects the investor's uncertainty about her future preference due to the randomness of the financial market states.
However, due to the  epistemic limitation or limited
information, an investor may have ambiguity about the driving force of market randomness and her future preference.
Focusing on such ambiguities, we will  introduce  a robust   forward performance measure, and investigate the corresponding portfolio selection problems.

The mean return rate and volatility are important factors characterizing the dynamics of risky assets.
In the traditional portfolio theory, these two factors are usually modeled by stochastic processes, the distributions of which are known to the decision-maker at each time node before the specified investment horizon.
In this case,
the investor is actually assumed to have full information on the   driving force of market randomness, and can
accurately  assigns probabilities to the various possible outcomes of investment or factors associated with the investment.
However,
in so complicated financial markets,
it is \emph{unrealistic} for investors to have accurate information on   the   dynamics or distributions of  the risk factors, essentially due to the cognitive limitation  on the  driving force of market randomness.
This situation  is referred to as  ``ambiguity" in the sense of Knight, while ``risk" in the former situation.
Ambiguity has raised  researchers' attention in the area of asset pricing   and  portfolio management \citep[see e.g.][]{Maenhout2004,Garlappi2007,wang2009optimal,Bossaerts2010,Liu2011a,Chen2014i,Luo2014,Guidolin2016,luo2016robustly,zeng2018ambiguity,Escobar2018}.
 We assume that an investor has ambiguous  beliefs on the paths of the risky asset price.
 Ambiguous beliefs are characterized by a set $\mathcal{P}$ of probability measures $(\mathbb{P}\in \mathcal{P})$ defined on the   canonical space $\Omega$, the set of continuous paths starting from the current price of the risky asset.
 We incorporate the investor's ambiguity on the risky asset price into her preference, by defining the forward performance measure on the canonical   space $\Omega$.

 We first characterize ambiguity on the dynamics of risky asset in terms of ambiguity on its mean return and volatility.
 More specifically, we assume that the mean return and the volatility processes of the risky asset lie in a convex compact set $\Theta\subset \mathbb{R}^2_+$, which then leads  to the set of probability measures $\mathcal{P}$.
%
  This formulation is different from the stochastic models with the known distributions at each time node, and generalizes the framework defined on a probability space with only one probability measure.
 Within in this general setting, we investigate an ambiguity-averse investor's  investment strategy,
  and her  conservative beliefs on the mean return and the volatility of risky assets.

 We then define the robust forward performance measure, by taking the investor's ambiguity on the deriving force of market randomness.
 In turn, we propose a method to construct such robust forward performance measure for a given initial preference, and derive the corresponding investment strategy and conservative beliefs on the mean return and the volatility of risky assets.
 We show that the sum of the market risk premium and the utility risk premium determines the trading direction.
  We further specify the initial preference of the constant relative risk aversion (CRRA) type, and investigate the determinants of the conservative beliefs on the mean return and the volatility of risky assets in three settings, i.e., ambiguity on the mean return rate, ambiguity on the volatility, and the structured ambiguity.
  When we consider ambiguity on the mean return rate, we keep the volatility as a constant, and  vise versa.
  Such ambiguities have been investigated in Merton's framework \citep[see e.g.][]{Lin2014c,luo2016robustly}.
  The third setting is motivated by the fact that there is no consensus on the relation between the mean return and the volatility of risky assets in  the empirical literature \citep[see e.g.][]{Omori2007,Bandi2012,Yu2012}, and investigated by \cite{Epstein2013}.
  We show that
  the sign of the total risk  premium determines the conservative belief on the mean return in the first setting, while the risk attitude and the relative value of the market risk premium over the utility risk premium jointly determine   the conservative belief on the volatility in the second setting.
  In the third setting,
  we would not derive the closed-form formula for the conservative beliefs, but show that the corresponding beliefs can take some intermediate value within the candidate value interval, as well as the upper and lower bounds.
  To our knowledge,
such interesting results are new   in the portfolio selection literature.

This paper contributes to the existing  literature in three folds.
\emph{First}, we propose a generic formulation of robust forward performance accommodating an investor's ambiguity on the dynamics of risky assets.
\emph{Second}, we figure out the determinants of trading direction for an investor in a market with one risk-free asset and one risky asset.
From the economic point of view,
it is the sum of the market risk premium and the utility risk premium  that determines an investor's trading direction.
\emph{Third}, we show that the market risk premium, the utility risk premium, and the risk tolerance affect an investor's conservative belief  on the mean return and volatility.
Especially, if the maximum of the total risk premium is negative, an investor will take the maximum of the mean return as the
worst-case value; if the minimum of the total risk premium is positive, an investor will take the minimum of the mean return as the value in the worst-case scenario; otherwise, the worst-case mean return lies between its minimum and maximum.
The market risk premium, the utility risk premium, and the risk tolerance jointly determine an investor's conservative belief on the volatility of risky assets.
We emphasize that
the conservative belief is related to the optimization associated with risk premiums, and these conservative beliefs may be   some intermediate values within their candidate value intervals, as well as boundaries.

\textbf{Related Literature}.
Most of the existing results on forward performance measures have so far focused on its construction and portfolio problems in the setting of risk, rather than    ambiguity  \cite[][to name a few]{Zariphopoulou2010,Musiela2010,Alghalith2012,Karoui2013,Kohlmann2013,Anthropelos2014,Nadtochiy2017,Avanesyan2018,Shkolnikov2015a,Case2018}.
As one of the few exceptions,
\cite{Kallblad2013a} investigate the robust forward performance measure in the setting of ambiguity characterized by a  set of equivalent probability measures.
However, this approach fails to solve the robust ``forward" investment problem under ambiguous volatility, since volatility ambiguity is characterized by a set of mutually singular probability measures \citep{Epstein2013}.
We fill this gap by characterizing an investor's ambiguity  with a set of probability measures, which may not be equivalent with each other.
Similar to our work, \cite{Chong2018} investigate robust forward investment under parameter uncertainty
in the framework where  a unique probability measure is aligned to the canonical space $(\Omega)$.
Different from such model setup,
we align a set of probability measures on the canonical space $(\Omega)$, accounting for an investor's ambiguity on the future scenarios of the risky asset price.
This approach is not only  technically  more general than the approach with a set of dynamic models under  a unique probability measure (as detailed in Remark 4 by \cite{Epstein2013}), but also allows an investor to explicitly incorporate  ambiguity on the risk source into her preference.
That is the key difference between our framework and the framework of
\cite{Chong2018}.
On the other hand, \cite{Chong2018} construct the forward performance measure  based on the solution of an infinite horizon  backward stochastic differential equation (BSDE).
Our approach associates the forward performance measure with a stochastic partial differential equation (SPDE), which provides the analogue of  the  Hamilton-Jacobi-Bellman equation (HJB) in   Merton's framework.
For the reason of tractability, we limit ourself to forward performance measures of some special forms, and investigate the corresponding robust investment.
It is out of this paper's scope to investigate the existence, uniqueness, and regularity    of the solution of  the associated SPDE in the general setting.
 Such simplified model setup and the corresponding results shed  light  on how ambiguity-aversion investors dynamically revise  their preferences as the market involves.

The remainder of this paper is organized as follows.
Section \ref{Setup} introduces the model setup for robust forward investment.
The construction of the robust forward performance measure is investigated in Section \ref{Construction}.
In Section \ref{CRRA:case}, we study the conservative belief of an ambiguity-averse investor with preference of the constant relative risk aversion (CRRA) type.
Section \ref{Conclusion} concludes.

\section{Model setup}
\label{Setup}

We consider a financial market with two tradable assets:  the risk-free bond  and the risky asset.
The risk-free bond has a constant return rate $r$, i.e.,
\begin{equation}\label{RisklessBondPrice}
  \ud P_t=r P_t \ud t\,,
\end{equation}
where $P$ is bond price with $P_0=1$.

The risk asset price $S = (S_t)_{t\in[0,\infty)} $  is modelled by the canonical process of   $\Omega$, defined by
$$ \Omega  = \left\{ \omega  = {(\omega (t))_{t \in [0,\infty )}} \in C([0,\infty ),\mathbb{R}^+):\omega (0) = S_0\right\}, $$
where $S_0$ is the current price of the risky asset and $S_t(\omega)=\omega(t)$.
We equip $\Omega$
 with the uniform norm and the corresponding Borel $\sigma$-field $\mathcal{F}$.
  $\mathbb{F} =  (\mathcal{F}_t)_{t\in[0,\infty)} $ denotes the canonical filtration, i.e., the natural (raw) filtration generated by $S$.
Due to the complication of financial market and the  limitation of individual    cognitive ability,
an investor may have  ambiguous belief on the risky asset price, i.e., ambiguity on the mean return $(\mu) $ or volatility $(\sigma)$ in our model setup.
We assume that $(\mu_t,\sigma_t)$ can take any value within  a convex compact set $\Theta\subset\mathbb R^{2}_+$, but without additional information about their distributions   for any  time $t\in[0,\infty)$.
That is,   $\Theta$ represents  ambiguity on the return and volatility of the risky asset.
More explicitly, we characterize ambiguity by $\Gamma^\Theta$, defined by
 \begin{equation}\label{processGama}
   \Gamma^\Theta=\Big\{\theta \mid \theta=(\mu_t ,\sigma_t )_{t\ge 0} \mbox{ is an }   \mathcal
{F}\mbox{-progressively measurable process and } \ (\mu_t,\sigma_t)\in \Theta \mbox{ for any } t>0 \Big\}\,.
 \end{equation}
For   $\theta=(\mu_t,\sigma_t)_{t \ge 0}\in \Gamma^\Theta$,
let $\mathbb{P}^\theta$ be the probability measure on $(\Omega,\mathbb{F})$ such that
the following stochastic differential equation (SDE)
 \begin{equation}\label{SDE1}
   \ud S_t=S_t(\mu_t\ud t+\sigma_t\ud W_t^\theta)\,,
\end{equation}
admits a unique  strong solution $S=(S_t)_{t\ge0}$,
where   $W^\mathbb{\theta}=(W^\mathbb{\theta}_t)_{t\ge 0}$ is a Brownian motion under $\mathbb{P}^\theta$.
   Let $\mathcal{P}^\Theta$  denote the set of   probabilities $\mathbb{P}^{\theta}$ on $(\Omega,\mathbb{F})$ such that the SDE \eqref{SDE1} has a unique strong solution, corresponding to the ambiguity characteristic  $\Theta$ $(\theta\in \Gamma^\Theta)$.
The Brownian motion $W^\theta$ can be interpreted as the driving force of   randomness behind the risky asset under the probability measure $\mathbb{P}^\theta$.
Such         model setup allows us to analyze how the investor's belief on the risky asset affects her preference and investment strategy, especially the effect of ambiguity on the risk source.

%
%
%
%
%
%
%

An investor is  endowed with some wealth $x_0>0$ at time $t=0$,  and allocates her wealth dynamically between the risky asset and the risk-free bond.
For $t\ge0$ and $s\ge t$, let  $\pi_s $ be the proportion of her wealth invested in the stock at time $s\ge 0$.
Due to the self-financing property, the discounted wealth $X^\pi=(X_s^\pi)_{s\ge t}$ is given by
\begin{equation}\label{Wealth}
  \ud X_s^\pi=(\mu_s-r)\pi_s X_s^\pi\ud s+\pi_s X_s^\pi \sigma_s\ud W_s^\mathbb{\theta},\quad X_t^\pi=x\,,
\end{equation}
where $r$ is the risk-free interest rate, $ W^\theta$ is a Brownian motion under $\mathbb{P}^\theta\in \mathcal{P}^\Theta$.
 The set of admissible strategies
 $\mathcal{A}(x)$
 is defined by
 \[\mathcal{A}(x)=\left\{\pi \left|  \pi  \mbox{ is self-financing and }  \mathbb{F}\mbox{-adapted}, \int_{t}^{s}\pi_r^2\ud r <\infty,  s\geq t, \;\right.\mathbb{P}^\theta\mbox{-a.s., for all}\; \mathbb{P}^\theta\in   \mathcal{P}^\Theta\right\}\,.\]

The optimal investment strategy $\pi^*$ and the corresponding wealth process $X^*$ are usually associated with an optimization problem, such as utility maximization or risk minimization.
%
Within Merton's framework for portfolio theory  \citep{Merton1969},
the value process $U(t,x;\tilde T)$ is formulated as
\begin{equation}\label{classical:U}
  U(t ,x;\tilde T):=\sup_{\pi\in \mathcal{A}_{\tilde T}}E[u(X_{\tilde T}^{\pi})\mid\F_t,\,X_t^\pi=x]\,,
\end{equation}
where the investment horizon $\tilde T$ is predetermined, $u$ is a utility function,
$\mathcal{A}_{\tilde T}$ is the set of admissible strategy, and $X_{\tilde T}^\pi$ is the terminal wealth corresponding to an admissible strategy $\pi\in \mathcal{A}_{\tilde T}$.
The expectation $(E)$ in \eqref{classical:U} is taken under some probability measure $\IP$\,,
if there is no ambiguity on the deriving force of market randomness.
Then, the dynamic programming principle can be applied to solve the optimal control problem \eqref{classical:U}, namely,
 \begin{equation}\label{DPP}
   U(t,x;\tilde  T)=\sup_{\pi\in \mathcal{A}_{\tilde T}}E[U(s,X_s^{\pi};\tilde T)\mid\F_t,\,X_t^\pi=x]\,.
 \end{equation}
By verification arguments, $U$ is the solution of  the Hamilton-Jacobi-Bellman (HJB) equation  \citep{Zhou1997}.
 The dynamic programming equation \eqref{DPP} essentially signifies that $\{U(t,X_t^{\pi};\tilde{T})\}_{t\in[0,\tilde T]}$ is a martingale at the optimum, and a supermartingale otherwise, associated with some probability measure $\IP$\,.
 This property can be interpreted as follows: if the system is currently at an optimal state, one needs to seek for controls which preserve   the same level of the average performance over all future times before the predetermined investment horizon $\tilde T$.
 We refer to this property as the martingale property of the value process.
On the other hand, \eqref{DPP} hints that $U(\tilde T,X_{\tilde T}^{\pi};\tilde{T})$ coincides with $u(X_{\tilde T}^{\pi})$, where $u$ represents the preference at $t=\tilde T$.
Note that the future utility function $u$  is specified at $t=0$.
However, it is not    intuitive   to specify the \emph{future} preference at the \emph{initial} time with complete isolation from the evolution of the market.
\cite{Musiela2007,Musiela2008} propose the so-called forward performance measure $U(t,x)$ which keeps the martingale property of $\{U(t,X_t^{\pi})\}_{t\in[0,T]}$ for \emph{any} horizon $T>0$, and coincides with the initial preference, namely $U(0,X_0^{\pi})=u(X_0^{\pi})$.
In this framework, the future preference  dynamically changes in accordance with the market evolution.

In the similar spirit of  \cite{Musiela2007,Musiela2008},
we  will generalize the definition of forward performance measure by considering  ambiguity on the risk source.
For $\mathbb{P}^\theta\in \mathcal{P}^\Theta$, we define  $\mathcal{P}^\Theta(t,\mathbb{P}^\theta)$
\begin{equation}\label{Con:P}
  \mathcal{P}^\Theta(t,\mathbb{P}^\theta):=\{\mathbb{P}'\in \mathcal{P}^\Theta\mid \mathbb{P}'=\mathbb{P}^\theta\mbox{ on } \mathcal{F}_t\}
\end{equation}
 which   facilitates the definition of the robust forward performance measure.

\begin{defi}[Robust forward performance]\label{RFP}

An $\F_t$-progressively measurable process $(U(t,x))_{t\ge0}$ is called a robust forward performance  if for $t\ge0$ and $x\in\R^+$, the following holds.
\begin{enumerate}[(i)]
  \item  The mapping $x\rightarrow U(t,x)$ is strictly concave and increasing.
  \item  For each $\pi\in \mathcal{A}(x)$,
      $\essinf_{\mathbb{P}\in \mathcal{P}^\Theta}\mathbb{E}^\mathbb{P} [U(t,X_t^\pi)]^+<\infty$, and
      \[\essinf_{\mathbb{P}'\in \mathcal{P}^\Theta(t,\mathbb{P}^\theta)}\mathbb{E}^\mathbb{P' } [U(s,X_s^\pi)\mid \F_t]\le U(t,X_t^\pi),\ t\le s, \; \mathbb{P}^\theta\mbox{-a.s.}\]

      \item
      There exists $\pi^*\in\mathcal{A}(x) $ for which
      \[  {\essinf_{\mathbb{P'}\in \mathcal{P}^\Theta(t,\mathbb{P}^\theta)}\mathbb{E}^\mathbb{P'} [U(s,X_s^{\pi^*})\mid \F_t]= U(t,X_t^{\pi^*}),\ t\le s,}\;\mathbb{P}^\theta\mbox{-a.s.}\]

\end{enumerate}
\end{defi}

Given the dynamics of the  forward performance measure $(U(t,x))_{t\ge0}$, we will solve the problem for optimal investment strategy, which can be formulated as a similar  problem as \eqref{DPP}.
\begin{Problem}[Robust Investment Problem]\label{RobustOpt}
  Given the robust forward performance $(U(t,x))_{t\ge0}$, the   investment problem is to solve
  \begin{equation}
   U(t,x) = \mathop {\sup }\limits_{\pi  \in \mathcal{A}(x)} \mathop {\inf }\limits_{\mathbb{P'}\in \mathcal{P}^\Theta(t,\mathbb{P}^\theta)} {\mathbb{E}^\mathbb{P'}}[U(s,X_s^\pi )\mid \mathcal{F}_t, {X_t^\pi} = x], \quad\mathbb{P}^\theta\mbox{-a.s.}\,,
 \end{equation}
 where $X^\pi$ follows \eqref{Wealth} and $\mathcal{P}^\Theta(t,\mathbb{P}^\theta)$ is given in \eqref{Con:P}.
\end{Problem}

 %
%

 The solution of this problem provides the robust investment strategy
$\pi^*$ and  the worst-case scenario  of $(\mu^{\pi^*},\sigma^{\pi^*})$ under ambiguity.
In turn, they will implicitly  provide
the corresponding probability measure  $\mathbb{P}^{\theta^*}$.
In the next section,
we will introduce the construction methods for the forward  performance under ambiguity, and then solve the robust investment problem.

\section{Robust Investment under the  Forward Performance Measures}
\label{Construction}


The specification of a forward performance measure $(\bar U(t,x))_{t\ge0}$ can take the market state and investor's wealth level into account at time $t$. Mathematically, $(\bar U(t,x))_{t\ge0}$  is called stochastic flow, a stochastic process with space parameter.
It can be characterized by its drift random field and diffusion random field. Under certain regularity hypotheses \citep{Karoui2013}, it can be written in the integration  form
 \begin{equation}\label{FP}
    \bar U(t,x)=u(x)+\int_0^t\beta(s,x)\ud s+\int_0^t\gamma(s,x)\ud \bar{B}_s \,,
 \end{equation}
 where $\bar{B}$ is the standard  Brownian motion defined on some probability space, $u$ is the initial utility, and $\beta$ and $\gamma$ are the so-called drift random field and the diffusion random field, respectively.
 To guarantee a stochastic flow $(\bar U(t,x))_{t\ge0}$ satisfy the definition \ref{RFP}, its drift random filed $\beta$ and diffusion random field $\gamma$ should satisfy some structure.
 By exploring such structure,
\cite{Musiela2010a} constructed some examples of forward performance measures.
In this framework, the driving force of market randomness is modelled by the standard Brownian motion $\bar {B}$.
We will generalize such framework, and  account for the ambiguity on the driving force of market randomness or the risky asset price.
Different from the dynamics of the risky asset price \eqref{SDE1},
we give even more freedom to an investor's preference,
and  propose the robust forward performance measure of the following form,
\begin{equation}\label{U:flow}
 U(t,x) = u(x) + \int_0^t \left[{\beta (s,x)}  +  \delta(s,x)\mu_s+ \gamma (s,x)\sigma_s^2 \right]\ud{s}  + \int_0^t {\eta (s,x)\sigma_s\ud{W_s^\theta}} ,
\end{equation}
where $W^\theta$ a Brownian motion   under
$\mathbb{P}^\theta\in \mathcal{P}^\Theta$ and $\theta=(\mu_t ,\sigma_t )_{t\ge 0} \in \Gamma^\Theta$\,.
 The random field $(\beta,\delta,\gamma,\eta)$ characterizes an investor's attitudes toward wealth level, ambiguity, and market risk.
Especially, the volatility term $\eta(t,x)\sigma_t$ of the robust forward performance measure reflects the investor's ambiguity about her preferences in the future, and   is subject to her choice.
The BSDE-based approach proposed by
\cite{Chong2018}  captures the investor's concern on parameter uncertainty by the generator of the associated BSDE.
Different from this BSDE-based approach,
we explicitly embed such concern into the axiomatic formulation \eqref{U:flow}.

For any given robust forward preference of the form   \eqref{U:flow},  the   investment problem \ref{RobustOpt}   allows  an investor to maximize her utility under the worst-case scenario of $ (\mu_t ,\sigma_t )_{t\ge 0} \in \Gamma^\Theta$.
To make the investment problem tractable,
the forward performance measure is assumed to be  regular enough.
For this reason, we introduce the notation of $\mathcal{L}^{2}(\mathbb{P})$-smooth stochastic flow.

 \begin{defi}[$\mathcal{L}^{2}(\mathbb{P})$-Smooth Stochastic Flow]
    Let $F: \Omega\times[0,\infty)\times\R\rightarrow\R$ be a stochastic flow with spatial argument $x$ and   local characteristics $(\beta,\gamma)$, i.e.,
 \[F(t,x) = F(0,x) + \int_0^t {\beta(s,x)\ud s}   + \int_0^t {\gamma (s,x)\ud{B^\mathbb{P}_s}} ,\]
 where $B^\mathbb{P}$ is a Brownian motion defined on a
   filtered probability space $(\Omega,(\mathcal
{F}_{t})_{t\geq 0},\mathbb{P})$.
  $F$ is said to be $\mathcal{L}^{2}(\mathbb{P})$-smooth  or  belong to  $\mathcal{L}^{2}(\Omega,(\mathcal
{F}_{t})_{t\geq 0},\mathbb{P})$, if
  \begin{enumerate}
    \item [(i)] for $x\in \mathbb{R}$, $F(\cdot,x)$ is continuous; for each $t>0$, $F(t,\cdot)$ is a $C^2$-map from $\R$ to $\R$, $\mathbb{P}$-a.s.,
    \item  [(ii)]$\beta:\Omega\times[0,\infty)\times\R\rightarrow\R $
    and $\gamma:\Omega\times[0,\infty)\times\R\rightarrow\R^d $ are continuous process continuous in $(t,x)$ such that 
     \begin{itemize}
      \item [(a)] for each $t>0$, $\beta(t,\cdot),\gamma(t,\cdot)$ belong  to $C^1(\R)$, $\mathbb{P}$-a.s.;
      \item [(b)] for each $x\in\R$, $\beta(\cdot,x)$ and $\gamma(\cdot,x) $ are $\mathcal{F}$-adapted.
    \end{itemize}
  \end{enumerate}
 \end{defi}

For  $\mathbb{P^\theta}\in \mathcal{P}^\Theta$, we are ready to formulate the robust forward performance as a $\mathcal{L}^2(\mathbb{P^\theta})$-smooth stochastic flow
%
  \begin{equation}\label{RU}
 U(t,x) = u(x) + \int_0^t \left[{\beta (s,x)}  +  \delta(s,x)\mu_s+ \gamma (s,x)\sigma_s^2 \right]\ud{s}  + \int_0^t {\eta (s,x)\sigma_s\ud{W_s^\theta}} ,
\end{equation}
where $W^\theta$ a Brownian motion   under
$\mathbb{P}^\theta\in \mathcal{P}^\Theta$ and $\theta=(\mu_t ,\sigma_t )_{t\ge 0} \in \Gamma^\Theta$\,.
Its smoothness plays a key role
 to construct the robust forward performance measures by specifying the structure of $(\beta,\delta,\gamma,\eta)$.

\begin{lem}\label{AssumptionThm}
For $\mathbb{P^\theta}\in \mathcal{P}^\Theta$, let $U$ be a  $\mathcal{L}^2(\mathbb{P^\theta})$-smooth stochastic flow as defined in (\ref{RU}).
   Let us suppose that

  \begin{enumerate}
  \item[(i)] the mapping $x\rightarrow U(t,x)$ is strictly concave and increasing;
    \item [(ii)]for an arbitrary  $\pi\in \mathcal{A}(x)$,  there exists $(\mu^\pi,\sigma^\pi)\in \Gamma^\Theta$, such that
  \[\essinf_{ \mathbb{P}'\in \mathcal {P}(t,\mathbb{P}^\theta)}\mathbb{E}^\mathbb{P'}[U(s,X_s^{\pi})\mid \F_t]= \mathbb{E}^{\mathbb{P}^{\mu^\pi,\sigma^\pi}} [U(s,X_s^{\pi})\mid \F_t], \;  t\le s, \mathbb{P^{\theta}}\mbox{-}a.s.\,,\]
  and
  $Y_s=U(s,X_s^{\pi})$ is a $\mathbb{P}^{\mu^\pi,\sigma^\pi}$-supermartingale;
    \item [(iii)] there exists $\pi^*\in \mathcal{A}(x)$ such that
    $Y_s^*=U(s,X_s^{\pi*})$ is a $\mathbb{P}^{ \mu^{\pi^*},\sigma^{\pi^*}}$-martingale.

  \end{enumerate}
Then $\pi^*$ is the optimal investment strategy for Problem \ref{RobustOpt}, associated with the worst-case scenario $(\mu^{\pi^*},\sigma^{\pi^*})$ of  $(\mu ,\sigma )$ .
\end{lem}

\begin{proof}
  For each $\pi\in \mathcal{A}(x)$,  since $Y_s=U(s,X_s^{\pi})$ is a $\mathbb{P}^{\mu^\pi,\sigma^\pi}$-supermartingale,
      \[\essinf_{ \mathbb{P}'\in \mathcal {P}(t,\mathbb{P}^\theta)}\mathbb{E}^\mathbb{P'}[U(s,X_s^{\pi})\mid \F_t]= \mathbb{E}^{\mathbb{P}^{\mu^\pi,\sigma^\pi}} [U(s,X_s^{\pi})\mid \F_t]\le U(t,X_t^\pi),\;   t\le s, \mathbb{P^{\theta}}\mbox{-}a.s. \]

     Since there exists $\pi^*\in \mathcal{A}(x)$ such that     $Y_s^*=U(s,X_s^{\pi*})$ is a $\mathbb{P}^{ \mu^{\pi^*},\sigma^{\pi^*}}$-martingale, we have
        \[\essinf_{ \mathbb{P}'\in \mathcal {P}(t,\mathbb{P}^\theta)}\mathbb{E}^\mathbb{P'}[U(s,X_s^{\pi})\mid \F_t]= \mathbb{E}^{\mathbb{P}^{\mu^{\pi^{*}},\sigma^{\pi^{*}}}} [U(s,X_s^{\pi^{*}})\mid \F_t]= U(t,X_t^{\pi^{*}}),\;  t\le s, \mathbb{P^{\theta}}\mbox{-}a.s. \]

Recalling the definition of robust forward performance (Definition \ref{RFP}),
we can see that $U$ is a forward performance, and the statement of this theorem is proved.
\end{proof}

Lemma \ref{AssumptionThm} provides a method to find the worst-case scenario of  the mean return and volatility of the risky asset, and the corresponding investment strategy, as stated in Theorem \ref{th4.1}.

\begin{thm} \label{th4.1}
 Let $U$ be a $\mathcal{L}^2(\mathbb{P^\theta})$-smooth stochastic flow on $(\Omega,\mathbb{F},\mathbb{P}^\theta)$ with $\mathbb{P^\theta}\in \mathcal{P}^\Theta$ and $\theta=(\mu_t ,\sigma_t )_{t\ge 0} \in \Gamma^\Theta$, and the mapping $x\rightarrow U(t,x)$ is strictly concave and increasing.
We suppose the following holds.
\begin{enumerate}
\item [(i)]$U$ satisfies  the following  equation
\begin{equation}\label{H1}
    \begin{aligned}
      &\sup_{\pi}\inf_{(\mu,\sigma)\in \Theta}\Big\{\beta(t,{x})+\delta(t,x)\mu+\gamma(t,x)\sigma^2+U_{x}(t,x)(\mu-r)\pi x\\
      &\quad\quad\quad\quad\quad+\eta_x(t,{x})\pi x\sigma^2+\frac{1}{2}U_{xx}(t,x)\pi^2\sigma^2x^2\Big\}=0.
    \end{aligned}
\end{equation}
  \item [(ii)] For any $\pi (t,x) \in \mathbb{R}$, there exists  $(\tilde{\mu}_t,\tilde{\sigma}_t)\in  \Theta$  such that
\begin{eqnarray*}\label{}
&& \inf_{(\mu,\sigma)\in \Theta}\Big\{\delta(t,x)\mu+\gamma(t,x)\sigma^2+U_{x}(t,x)(\mu-r)\pi x+\eta_x(t,{x})\pi x\sigma^2+\frac{1}{2}U_{xx}(t,x)\pi^2\sigma^2x^2\Big\}\\
&&=\delta(t,x)\tilde{\mu}_t+\gamma(t,x)\tilde{\sigma}^2_t+U_{x}(t,x)(\tilde{\mu}_t-r)\pi x+\eta_x(t,{x})\pi x\tilde{\sigma}^2_t+\frac{1}{2}U_{xx}(t,x)\pi^2\tilde{\sigma}^2_tx^2.
\end{eqnarray*}
\end{enumerate}

Let $\pi^{*}(t,x) \in \mathbb{R}$ satisfy
\begin{eqnarray*}\label{}
   \pi^{*}(t,x) &=& \arg\sup\limits_{\pi} \inf_{(\mu,\sigma)\in \Theta}\Big\{\delta(t,x)\mu+\gamma(t,x)\sigma^2+U_{x}(t,x)(\mu-r)\pi x+\eta_x(t,{x})\pi x\sigma^2\\&&\qquad+\frac{1}{2}U_{xx}(t,x)\pi^2\sigma^2x^2\Big\},
\end{eqnarray*}
and $(\mu^{*},\sigma^{*})$ satisfy
\begin{eqnarray}\label{martingale}
   &&\sup\limits_{\pi}\inf\limits_{(\mu,\sigma)\in \Theta}\Big\{ \delta(t,x)\mu+\gamma(t,x)\sigma^2+U_{x}(t,x)(\mu-r)\pi x+\eta_x(t,{x})\pi x\sigma^2+\frac{1}{2}U_{xx}(t,x)\pi^2\sigma^2x^2\Big\}\nonumber\\
&=  & \delta(t,x)\mu^{*}_t+\gamma(t,x)(\sigma^{*}_t)^2+U_{x}(t,x)(\mu^{*}_t-r)\pi^{*} x+\eta_x(t,{x})\pi^{*} x(\sigma^{*}_t)^2\nonumber\\&&+\frac{1}{2}U_{xx}(t,x)(\pi^{*})^2(\sigma^{*}_t)^2x^2.
\end{eqnarray}

Let $X^{*}$ be the unique solution of the stochastic differential equation
\begin{equation*}\label{X:star}
  \ud X^{*}_t=(\mu^{*}_t-r)\pi^{*}_t X_t^* \ud t+\pi^{*}_t X_t^* \sigma^{*}_t\ud W_t^\mathbb{\mu^{*},\sigma^{*}},\quad X^{*}_0 =x.
\end{equation*}
Then $\pi^{*}(t,X_{t}^{*} ) $ solves the Problem \ref{RobustOpt}.
\end{thm}

\begin{proof}
  Under the regularity conditions on $U$, we apply the It\^{o}-Ventzell formula to $U(t,X^\pi)$ for any admissible portfolio $X^\pi$ under each $\mathbb{P}^\theta\in \mathcal{P}^\Theta$
  \begin{align*}
   \ud U(t,{X_t^\pi}) &=   \left\{\beta (t,{X_t^\pi})+\delta(t,X_t^\pi)\mu_t+\gamma(t,X^\pi_t)\sigma_t^2\right\} \ud t   +   {\eta (t,{X_t^\pi})\sigma_t\ud{W^\theta_t}}  +  {{U_x}(t,{X_t^\pi})\ud{X_t^\pi}}\\
   & \quad+ \frac{1}{2}  {{U_{xx}}(t,{X_t^\pi})\ud \langle {X^\pi} \rangle_t }    +   {{\eta _x}(t,{X_t^\pi})\sigma_t\ud  \langle X^\pi,{W^\theta}{ \rangle _t}}\\
   &=\Big\{\beta(t,{X_t^\pi})+\delta(t,X_t^\pi)\mu_t+\gamma(t,X^\pi_t)\sigma_t^2+U_{x}(t,X_t^\pi)(\mu_t-r)\pi_t X_t^\pi+\eta_x(t,{X_t^\pi})\pi_t{X_t^\pi}\sigma_t^2 \\
   &\quad +\frac{1}{2}U_{xx}(t,X_t^\pi)\pi^2_t\sigma_t^2({X_t^\pi})^2\Big\}\ud t+\left\{\eta(t,X_t^\pi)\sigma_t+U_x(t,X_t^\pi)\pi_t X_t^\pi\sigma_t\right\}\ud W_t^\mathbb{\theta}\,.
  \end{align*}

We denote by $g(t,\mu_t, \sigma_t)=\beta(t,{X_t^\pi})+\delta(t,X_t^\pi)\mu_t+\gamma(t,X^\pi_t)\sigma_t^2+U_{x}(t,X_t^\pi)(\mu_t-r)\pi_t X^\pi+\eta_x(t,{X_t^\pi})\pi_t{X_t^\pi}\sigma_t^2 +\frac{1}{2}U_{xx}(t,X_t^\pi)\pi^2_t\sigma_t^2({X_t^\pi})^2$.

For $t<s$,
 \begin{eqnarray*}
 \essinf_{\mathbb{P}'\in \mathcal {P}(t,\mathbb{P}^\theta)}\mathbb{E}^\mathbb{P'}[U(s,X_s^{\pi})\mid \F_t]
 &= &\essinf_{\mathbb{P'}\in \mathcal {P}(t,\mathbb{P}^\theta)}\mathbb{E}^\mathbb{P'}[U(t,X_t^{\pi})+\int_{t}^{s}g(r,\mu_r, \sigma_r)\ud r\mid \F_t]\\
 &\geq&\essinf_{\mathbb{P}'\in \mathcal {P}(t,\mathbb{P}^\theta)}\mathbb{E}^\mathbb{P'}[U(t,X_t^{\pi})+\int_{t}^{s}\inf\limits_{\mu, \sigma}g(r,\mu , \sigma )\ud r\mid \F_t]\\
 &=&\essinf_{\mathbb{P}'\in \mathcal {P}(t,\mathbb{P}^\theta)}\mathbb{E}^\mathbb{P'}[U(t,X_t^{\pi})+\int_{t}^{s}g(r,\tilde{\mu} , \tilde{\sigma} )\ud r\mid \F_t]\\
 &=&\mathbb{E}^\mathbb{P^{\mu^\pi,\sigma^\pi}}[U(t,X_t^{\pi})+\int_{t}^{s}g(r,\tilde{\mu} , \tilde{\sigma} )\ud r\mid \F_t]\\
  &=&\mathbb{E}^\mathbb{P^{\mu^\pi,\sigma^\pi}}[U(s,X_s^{\pi})\mid \F_t]\,.
 \end{eqnarray*}
where $(\mu^{\pi},  \sigma^{\pi})=(\mu, \sigma)$ on $[0,t]$, and $(\mu^{\pi},  \sigma^{\pi})=(\tilde{\mu}, \tilde{\sigma})$ on $[t,s]$.

Therefore,
\begin{equation}\label{Cor:Lem3}
  \essinf_{\mathbb{P}'\in \mathcal {P}(t,P^\theta)}\mathbb{E}^\mathbb{P'}[U(s,X_s^{\pi})\mid \F_t]
  = \mathbb{E}^\mathbb{P^{\mu^\pi,\sigma^\pi}}[U(s,X_s^{\pi})\mid \F_t]\,.
\end{equation}

 It is obvious that $U(s,X_s^{\pi})$ is a $\mathbb{P}^{\theta}$-supermartingale. From (\ref{martingale}) it follows that
    $U(s,X_s^{\pi^*})$ is a $\mathbb{P}^{ \theta^*}$-martingale.
    Recalling Lemma \ref{AssumptionThm} and \eqref{Cor:Lem3},
    $(\mu^*,\sigma^*)$ represents the worst-case scenario of the mean return and volatility of the risky asset, and $\pi^*$ is the corresponding investment strategy.
\end{proof}

Theorem \ref{th4.1} provides a natural way to construct a robust forward performance measure, optimal investment strategy and the worst-case scenario of the mean return and volatility of risky assets.
We summarize such results in Corollary \ref{Col3}.

 \begin{coll}\label{Col3}

\begin{enumerate}
  \item [(i)] If the $U$ is a robust forward performance measure and the worst-case $(\mu^*,\sigma^*)$ is selected,
the optimal investment strategy is given in the feedback form
 \begin{equation}\label{Strategy:vol}
   \tilde \pi (t,x) =  - \frac{{{\eta _x}(t,x){{\sigma_{t}^*} ^2} + (\mu_{t}^*  - r){U_x}(t,x)}}{{x{{\sigma_{t}^*}^2}{U_{xx}}(t,x)}}\,,
 \end{equation}
 where the first and second term of the optimal strategy will be referred to as its non-myopic and  myopic part, respectively \citep{Musiela2010a}.
  \item [(ii)] If the $U$ is a robust forward performance measure, its characteristics $(\beta,\delta,\gamma,\eta)$
      should satisfy.
      \begin{equation} \label{con}
        \inf_{(\mu,\sigma)\in \Theta}\left\{\beta+\delta\mu+\left(\gamma-\frac{\eta_x^2}{2U_{xx}(t,x)}\right)\sigma^2
        -\frac{(\mu-r)^2U_x^2(t,x)}{2U_{xx}(t,x)\sigma^2}
        -\frac{(\mu-r)U_x\eta_x}{U_{xx}(t,x)}\right\}=0\,,
      \end{equation}
      for   $(t,x)\in [0,\infty)\times\R^+$.
     The solution of condition \eqref{con} leads to the worst-case $(\mu^*,\sigma^*)$.
\end{enumerate}

\end{coll}
The constraint \eqref{con} on the local characteristics   $(\beta,\delta,\gamma,\eta)$ implies that the
forward performance measure is not unique for a given initial utility function. By specifying three of them,
we can calculate the fourth one.
Hence, the investor in this framework has the freedom to  specify her initial utility, as well as the  additional characteristics of the utility field.
However, in Merton's framework, the dynamics and characteristics of the utility field are derived   from the terminal utility function, which is specified by the investor  at the initial time.
We note that the constraint \eqref{con} holds in the path-wise sense.

The  local characteristics  $(\beta,\delta,\gamma,\eta)$  actually can be used to represent the investor's    attitude through
local risk tolerance $\tau^U(t,x)=-\frac{U_x(t,x)}{U_{xx}(t,x)}$, utility risk premium $\varrho^U(t,x,\sigma)=\frac{\eta_x(t,x)\sigma_t}{U_x(t,x)}$ \citep{Karoui2013},
  and market risk premium $m(\mu,\sigma)=\frac{\mu -r}{\sigma }$.
Actually,
the optimal strategy $\tilde \pi$ \eqref{Strategy:vol} can be written as
\begin{align}
  \tilde \pi (t,x)& =\frac{\mu^*-r}{x{\sigma^*}^2}\tau^U-\frac{\eta_x(t,x)}{xU_{xx}(t,x)} \label{Strategy:vol22}\\
 &=\frac{\tau^U}{\sigma^*x}\left(\frac{\mu^*-r}{\sigma^*}+ \frac{\eta_x(t,x)\sigma^*}{U_x(t,x)} \right)=\frac{\tau^U}{\sigma^*x}\left(m(\mu^*,\sigma^*)+\varrho^U(t,x,\sigma^*) \right)\,.\label{Strategy:vol2}
\end{align}
 The first component of the investment strategy \eqref{Strategy:vol22}, known  as myopic strategy,
 resembles the
investment policy followed by an investor in markets in which the investment opportunity set remains constant through time.
 The second one is called the
excess hedging demand and represents the additional (positive or negative) investment generated by the volatility process $\eta\sigma$ of the performance process $U$  \citep{Musiela2010a}.
 Essentially, the investment strategy \eqref{Strategy:vol2} reveals that it is affected by the investor's risk tolerance, market risk premium, and utility risk premium, as well as the worst-case scenario of  the mean return $\mu$ and the volatility $\sigma$ of the risky asset.
 Obviously, it is   the sum of utility risk premium and market risk premium that determines the trading direction of an investor.
 Such statement holds regardless of the specification    of the robust forward performance measure.
Note that  the worst-case scenario of $(\mu,\sigma)$ is characterized by \eqref{con}.
%
%
To analyze the implication of ambiguity,
we restrict ourself to the robust forward performance measure of special forms, and derive the analytical solution for \eqref{con}.

\section{Robust Forward Performance  of the CRRA type}
\label{CRRA:case}

Utility function of the CRRA type is one of the commonly used utility function, which is a power function of wealth.
We assume an investor's  dynamic preference is characterized by utility function of the CRRA type over the time $t\in[0,\infty)$, with the initial  utility function $u(x)=x^\kappa/\kappa$, $\kappa\in(0,1)$ and time-varying coefficients.
More specifically, we set such   forward performance $U$   of  the following form
\begin{equation}\label{U:pow}
 \left\{
 \begin{aligned}
   U(t,x)&=\frac{ \exp(\alpha(t))}{\kappa}x^\kappa, &&U(0,x)=x^\kappa/\kappa\,,\\
   \ud \alpha(t)&=f(t)\ud t+g(t)\sigma_t\ud W_t^\theta,&&\alpha(0)=0\,,
 \end{aligned}
 \right.
\end{equation}
where  $\kappa\in(0,1)$ and $W^\theta=(W^{\theta}_t)_{t\ge 0}$ is a Brownian motion defined on a
   filtered probability space $(\Omega,\mathbb{F},\mathbb{P}^\theta)$ with $\mathbb{P^\theta}\in \mathcal{P}^\Theta$ and $\theta=(\mu,\sigma)\in \Gamma^\Theta$.
Without loss of generality.
Its differential form  is then given by
\begin{align}
  \ud U(t,x) & =U(t,x)\left(f(t)\ud t+\frac{1}{2}g^2(t)\sigma_t^2\ud t+g(t)\sigma_t\ud W_t^\theta\right),\quad U(0,x)=x^\kappa/\kappa\,, \label{Power:F}
\end{align}
and
\begin{align}
U_x(t,x) & = x^{\kappa-1}\exp(\alpha(t)),\label{Pow:F1}\\
    U_{xx}(t,x)&= (\kappa-1)x^{\kappa-2}\exp(\alpha(t)).
\end{align}
 In this case, the utility risk premium $\varrho^U(t,x,\sigma) =\varrho^U(\sigma)=g(t)\sigma$.

We can rewrite the forward performance measure \eqref{Power:F} in the form of \eqref{RU}, where
\begin{equation}\label{Pow:F2}
  \begin{aligned}
  \beta & =U(t,x)f(t)\,, & \gamma & =\frac{1}{2}U(t,x)g^2(t)\,,&\\
  \delta&=0\,,& \eta&=U(t,x)g(t)\,.&\\
\end{aligned}
\end{equation}
The characteristics $(\beta,\delta,\gamma,\eta)$
can be substituted into the constraints
 \eqref{con}, to specify the structure of the forward performance \eqref{U:pow}.

 If there is no ambiguity on the mean return and volatility, the constraint \eqref{con} is reduced to
 \begin{equation}\label{con2}
   f(t)=\frac{1}{2}\frac{ {g^2}(t){\sigma_t ^2}}{\kappa-1}{\rm{ + }}\frac{\kappa }{{ \kappa-1 }} \left\{ {\frac{1}{2}\frac{{{{(\mu_t  - r)}^2}}}{{{\sigma_t^2}}}{\rm{ + }}(\mu_t  - r)g(t)} \right\}\,,
 \end{equation}
  and the corresponding investment strategy is given by
  \begin{align}
    \pi^*&=\frac{g(t){\sigma_t}^2+(\mu_t-r)}{(1-\kappa){\sigma_t}^2}\notag\\
     & =\frac{1}{(1-\kappa)\sigma_t}\left(\frac{\mu_t-r}{\sigma_t}+g(t)\sigma_t \right)\notag\\
     & =\frac{1}{(1-\kappa)\sigma_t}\left(m(\mu_t,\sigma_t)+\varrho^U(\sigma_t)\right)\,.\label{Str}
  \end{align}
  The optimal investment strategy without ambiguity \eqref{Str}, as well as the optimal strategy with ambiguity \eqref{Strategy:vol2}, implies that the market risk premium and utility risk premium play an important role in the trading direction in both settings.

 In the following sub-sections, we will consider  an investor's conservative beliefs and the forward performance of the CRRA type in different settings: ambiguity on mean return $\mu$,
 ambiguity on the volatility $\sigma$, and ambiguity on both mean return and volatility.
 The structure of forward performance in these settings will involve   optimizations with respect to $\mu$ and $\sigma$, as implied by the constraint \eqref{con}.

\subsection{Ambiguity only on the mean return }
Ambiguity on the mean return is referred to as
  the case where the dynamics of mean return is ambiguous, with known dynamics of volatility.
  For the sake of simplicity, we assume $\sigma_t$ is  known as a  constant $\sigma$.

%
  \begin{prop} Assume an investor's forward preference $U$ is characterized by the initial  utility function $u(x)=x^\kappa/\kappa$  with $\kappa\in (0,1)$\,,  and propagates in the following form
  \begin{align}
  \ud U(t,x) & =U(t,x)\left(f(t)\ud t+\frac{1}{2}g^2(t)\sigma ^2\ud t+g(t)\sigma \ud W_t^\theta\right), \quad U(0,x)=u(x),\label{Power:F}
\end{align}
   where $f$ and $g$ are deterministic functions of $t$,   $\sigma$ is the volatility of the risky asset, and $W^\theta$ is a Brownian motion defined on a
   filtered probability space $(\Omega,\mathbb{F},\mathbb{P}^\theta)$ with $\mathbb{P^\theta}\in \mathcal{P}^\Theta$ and $\theta=(\mu,\sigma)\in \Gamma^\Theta$.

If the investor's ambiguity is characterized by the lower bound $\underline{\mu}$ and upper bound $\overline{\mu}$ of $\mu$, $f$ should satisfy the following condition
 \begin{equation}\label{Power:f3aa}
f (t)= \frac{1}{2}\frac{ {g^2}(t){\sigma ^2}}{\kappa-1}{\rm{ + }}\frac{\kappa }{{ \kappa-1 }} \left\{ {\frac{1}{2}\frac{{{{(\mu^*  - r)}^2}}}{{{\sigma^2}}}{\rm{ + }}(\mu^*  - r)g(t)} \right\}\,,
 \end{equation}
where
 \begin{equation}\label{Ustar}
   \mu^*=
   \left\{
   \begin{aligned}
     &\overline{\mu}\,, && \quad\mbox{ if } \;\frac{\overline{\mu}-r}{\sigma}<-g(t)\sigma  \,,&\\
    & r-g\sigma^2\,, & &\quad\mbox{ if }\;\frac{\underline{\mu}-r}{\sigma}\le -g(t)\sigma \le \frac{\overline{\mu}-r}{\sigma},\\
    & \underline{\mu}\,,&& \quad\mbox{ if }\;\frac{\underline{\mu}-r}{\sigma}\ge -g(t)\sigma \,.&
   \end{aligned}
   \right.
   \end{equation}
   Corresponding to the selection of worst-case mean return $\mu^*$, the investment strategy $\pi^*$ is given by
\begin{equation} \label{Pow:Stratergy:m}
  \pi^*=\frac{g(t){\sigma }^2+(\mu^*-r)}{(1-\kappa){\sigma }^2}\,.
\end{equation}
  \end{prop}

  \begin{proof}
     In this case, the constraint \eqref{con} is reduced to
  \begin{equation}\label{Power:f3a}
  f(t) = \frac{1}{2}\frac{ {g^2}(t){\sigma ^2}}{\kappa-1}{\rm{ + }}\kappa {\sup _{\mu  \in [\underline{\mu},\overline{\mu}]}}\left\{ {\frac{1}{2}\frac{{{{(\mu  - r)}^2}}}{{(\kappa  - 1)\sigma  ^2}}{\rm{ + }}\frac{{(\mu  - r)g(t)}}{{\kappa  - 1}}} \right\}\,,
 \end{equation}
 Assume the supermum is achieved at $\mu^*$.
 Simple calculations lead to
 \begin{equation} \label{UStare2}
   \mu^*=
   \left\{
   \begin{aligned}
     &\overline{\mu}\,, && \quad\mbox{ if } \;\overline{\mu}-r<-g(t)\sigma^2 \,,&\\
    & r-g\sigma^2\,, & &\quad\mbox{ if }\;\underline{\mu}-r\le -g(t)\sigma^2\le \overline{\mu}-r,\\
    & \underline{\mu}\,,&& \quad\mbox{ if }\;\underline{\mu}-r\ge -g(t)\sigma^2\,.&
   \end{aligned}
   \right.
   \end{equation}
   Due to $\sigma>0$, the belief on the worst-case return \eqref{UStare2} is equivalent to that given by \eqref{Ustar}.
   Correspondingly,
the optimal strategy \eqref{Strategy:vol} is   reduced to
  \begin{equation}
  \pi^*=\frac{g(t){\sigma}^2+(\mu^*-r)}{(1-\kappa){\sigma}^2}\,.
\end{equation}
  \end{proof}

We can  interpret the selection rule \eqref{Ustar} from the premium point of view.
Recalling the definition of the market risk premium  $m(\mu,\sigma)$ and the utility risk premium $\varrho^U(\sigma)$, i.e.,
 \[m(\mu,\sigma)=\frac{\mu-r}{\sigma} \quad\mbox{and}\quad\varrho^U(\sigma)=g(t)\sigma\,,\]
we can rewrite \eqref{Ustar} as
 \begin{equation}\label{Ustar2}
   \mu^*=
   \left\{
   \begin{aligned}
     &\overline{\mu}\,, && \quad\mbox{ if } \;m(\overline{\mu},\sigma)+\varrho^U(\sigma)<0 \,,&\\
    & r-g\sigma^2\,, & &\quad\mbox{ if }\;m(\underline{\mu},\sigma)+\varrho^U(\sigma)\le 0\le m(\overline{\mu},\sigma)+\varrho^U(\sigma) \,,&\\
    & \underline{\mu}\,,&& \quad\mbox{ if }\;m(\underline{\mu},\sigma)+\varrho^U(\sigma)>0 \,.&
   \end{aligned}
   \right.
   \end{equation}

It implies that the   worst-case mean return and the trading direction depend  on the  total risk premium that the investor can achieve in the setting of ambiguity on mean return, i.e., $ m(\mu,\sigma)+\varrho^U (\sigma)$\,.
When  $m(\underline{\mu},\sigma)+\varrho^U(\sigma)$ is positive, an investor will take $\underline{\mu}$ as the worst-case  mean return, and take a long position $(\pi>0)$.
When $m(\overline{\mu},\sigma)+\varrho^U(\sigma)$ is negative, an investor will take $\overline{\mu}$ as the worst case, and take a short position $(\pi<0)$.
Otherwise, she will take $r-g\sigma^2$ as the worst-case mean return, and do not invest on the risky asset $(\pi=0)$.
From this point of view,
it is the total risk premium that characterizes  the worst-case mean return and the investor's trading direction.

Such premium-based rule \eqref{Ustar} on the  conservative belief towards  the mean return is consistent with the rule proposed by
\cite{Chong2018} and \cite{Lin2014c}.  \cite{Chong2018}  propose  to select the worst-case scenario of the mean return in a feedback form associated with the position on risky assets, i.e.,
the long  and short positions correspond to $\underline{\mu}$ and $\overline{\mu}$, respectively.
In the classical framework, the selection of worst-case mean return dependents on the investor's position on the risky asset, as argued by \cite{Lin2014c} that
nature decides for a low drift if an investor takes a long position, and for a high drift if an investor takes a long position.
 However, the  rule \eqref{Ustar} is not given in a feedback form associated with an investor's position, but  directly related to   the market situations and the investor's utility risk premium.
%
In this new framework, we highlight the combination effect of the utility risk premium and the market risk premium on the worst-case   mean return of the risky asset.




\subsection{Ambiguity only on volatility}

We refer to volatility ambiguity as the case where the dynamics of volatility is unknown, but constrained in the interval $[\underline{\sigma},\overline{\sigma}]$ with $0< \underline{\sigma}\le\overline{\sigma}$.
For the sake of simplicity, we suppose  $\mu_t$ to be a constant $\mu$ over the time.

 \begin{prop} Assume an investor's preference $U$ is characterized by the initial  utility function $u(x)=x^\kappa/\kappa$  with $\kappa\in (0,1)$\,,  and propagates in the following form
  \begin{align}
  \ud U(t,x) & =U(t,x)\left(f(t)\ud t+\frac{1}{2}g^2(t)\sigma_t^2\ud t+g(t)\sigma_t\ud W_t^\theta\right), \quad U(0,x)=u(x),\label{Power:F}
\end{align}
   where $f$ and $g$ are deterministic functions of $t$,   $\sigma$ is the volatility of the risky asset, and $W^\theta$ is  a Brownian motion defined on the
   filtered probability space $(\Omega,\mathbb{F},\mathbb{P}^\theta)$  with $\mathbb{P^\theta}\in \mathcal{P}^\Theta$ and $\theta=(\mu,\sigma)\in \Gamma^\Theta$.

If the investor ambiguity is characterized by the lower bound $\underline{\sigma}$ and upper bound $\overline{\sigma}$ of $\sigma$,
$f$ should satisfy the following structure
 \begin{equation}\label{Power:fv}
  f (t)= \frac{{\kappa (\mu  - r)g(t)}}{{\kappa  - 1}} + \frac{1}{2(\kappa-1)}  {\left( {{g^2}(t){{\sigma^*} ^2} + \frac{{\kappa {{(\mu  - r)}^2}}}{{{{\sigma^*} ^2}}}} \right)}  \,,
 \end{equation}
 where
  \begin{equation}\label{Con:Pow}
   {\sigma^*}^2=\left\{
   \begin{aligned}
     &\underline{\sigma}^2,&&\mbox{ if } g^2(t)\ge \frac{\kappa(\mu-r)^2}{\underline{\sigma}^4},\\
   & \overline{\sigma}^2,&&\mbox{ if } g^2(t)\le \frac{\kappa(\mu-r)^2}{\overline{\sigma}^4},\\
   & \frac{|\mu-r|}{|g(t)|}\sqrt{\kappa},&&\mbox{ if } \frac{\kappa(\mu-r)^2}{\overline{\sigma}^4} \le g^2(t)\le \frac{\kappa(\mu-r)^2}{\underline{\sigma}^4}.
   \end{aligned}
   \right.
 \end{equation}
 Correspondingly, the optimal investment strategy is
\[ \pi^*  =\frac{g(t){\sigma^*}^2+(\mu-r)}{{\sigma^*}^2(1-\kappa)}\,.\]
\end{prop}
\begin{proof}
  In this setting, the constraint  \eqref{con} is reduced to
 \begin{equation}\label{Power:fv}
  f = \frac{{\kappa (\mu  - r)g(t)}}{{\kappa  - 1}} + {\sup _{{\sigma ^2\in[\underline{\sigma}^2,\overline{\sigma}^2]}}} \frac{1}{2( \kappa-1)}  {\left( {{g^2}(t){\sigma ^2} + \frac{{\kappa {{(\mu  - r)}^2}}}{{{\sigma ^2}}}} \right)}  \,.
 \end{equation}

 To solve the optimization problem, we denote $\theta:=\sigma^2$, and define a function $h$ by
 \[h(\theta)=-g^2(t)\theta-\frac{\kappa(\mu-r)^2}{\theta},\quad \theta\in[\underline{\sigma}^2,\overline{\sigma}^2]\,.\]
 Simply analysis can provide that $h$ reaches  its maximum at $\theta^*$, where
 \begin{equation}\label{Con:Pow}
  \theta^*=\left\{
   \begin{aligned}
     &\underline{\sigma}^2,&&\mbox{ if } g^2(t)\underline{\sigma}^2\ge \frac{\kappa(\mu-r)^2}{\underline{\sigma}^2},\\
   & \overline{\sigma}^2,&&\mbox{ if } g^2(t)\overline{\sigma}^2\le \frac{\kappa(\mu-r)^2}{\overline{\sigma}^2},\\
   & \frac{|\mu-r|}{|g(t)|}\sqrt{\kappa},&&\mbox{ if } \frac{\kappa(\mu-r)^2}{\overline{\sigma}^4} \le g^2(t)\le \frac{\kappa(\mu-r)^2}{\underline{\sigma}^4}.
   \end{aligned}
   \right.
 \end{equation}
 Due to $\theta=\sigma^2$, we have the worst-case scenario ${\sigma^*}^2$ of $\sigma^2$ \eqref{Con:Pow}.
\end{proof}

The conservative belief on the volatility depends on the market risk premium and the utility risk premium, as the case of the conservative belief on mean return \eqref{Ustar} or \eqref{Ustar2}.
We will show that it is the relative value of these two premiums that determines the conservative belief on volatility.
Note that it is  the sum of these two premium  determines the conservative belief on mean return, as shown by \eqref{Ustar} or \eqref{Ustar2}.
In our specific setting,   ambiguity on mean return only affects the market risk premium, while ambiguity on volatility affects both the market risk premium and the utility risk premium.
It is then natural to consider the effects of their relative value.

  Define the relative value $\tau(\sigma)$ of the utility risk premium  $\varrho^U(\sigma)$ with respect to the market price of risk $m(\mu,\sigma)$ as
  \[\tau(\sigma)=\frac{\varrho^U(\sigma)}{m(\mu,\sigma)}:=\frac{g(t)\sigma}{\frac{ \mu-r }{\sigma}}\,.\]
  Then, the   worst-case volatility \eqref{Con:Pow} can be rewritten as
    \begin{equation}\label{Con:Pow2}
   {\sigma_t^*}^2=\left\{
   \begin{aligned}
     &\underline{\sigma}^2,&&\mbox{ if } \tau^2(\underline{\sigma})\ge\kappa,\\
   & \overline{\sigma}^2,&&\mbox{ if } \tau^2(\overline{\sigma})\le\kappa ,\\
   & \frac{|\mu-r|}{|g(t)|}\sqrt{\kappa},&&\mbox{  otherwise} .
   \end{aligned}
   \right.
 \end{equation}
    The rule \eqref{Con:Pow2} for worst-case volatility implies  that
   if the relative value of  the utility risk premium over the market risk premium is large enough than the investor's risk-averse attitude $\kappa$,
   the investor will take $\underline{\sigma}$ as the worst-case volatility.
   Alternatively, if such relative value is smaller enough than the investor's risk-averse attitude, the investor will take $\overline{\sigma}$ as the worst-case volatility.
   Otherwise,
   the   worst-case volatility depends on her attitude toward risk and ambiguity about her future preferences.
   Overall,
   an ambiguity-averse investor will take her attitude toward risk and ambiguity into account when ambiguous on the volatility of the driving force of market randomness.



\subsection{Structured ambiguity on mean return and volatility}

Empirical research shows that the mean return can be either positively or negatively related to the volatility of risky assets \citep[see e.g.][]{Omori2007,Bandi2012,Yu2012}.
Without a consensus of their relation,
we employ a flexible model to capture the structured ambiguity on the mean return and volatility of the driving force of market randomness \citep{Epstein2013,Epstein2014},
\begin{equation}\label{structure}
 \Theta=\Big\{(\mu,\sigma)\mid \sigma^2  = \sigma_0^2+\alpha z, \mu-r  =\mu_0+z, z\in [z_{ 1},z_{2}]\Big\},
\end{equation}
where $\sigma_0,\mu_0>0$ and
$\alpha\in\R$ such that $\sigma^2>0$.
$\alpha>0$ implies that the return is positively related to the volatility, and vice versa.
The selection of worst-case value of mean return and volatility
will be reduced to the selection of $z^*\in [z_{ 1},z_{2}]$, where the spread $z_2-z_1$ represents the size of an investor's ambiguity on the mean return and volatility.

%
%
%


Recalling the constraints \eqref{con} and \eqref{Pow:F2},
  we have
  \begin{align}
    f &  = \sup _{\mu,\sigma^2}\frac{1}{\kappa-1}\left\{\kappa g(t)(\mu-r)+\frac{1}{2}\left(g^2(t)\sigma^2+\frac{\kappa(\mu-r)^2}{\sigma^2}\right)\right\}\notag\\
  &={\sup _{{z}}}\frac{1}{\kappa-1}\left\{  {{\kappa (\mu_0+z  )g(t)}}{ } +{\frac{1}{{2 }}\left( { {(\sigma_0 ^2+\alpha z)}{g^2(t)} + \frac{{\kappa {{(\mu_0  +z)}^2}}}{{{\sigma_0 ^2+\alpha z}}}} \right)} \right\}\notag\\
  &=  {\sup _{{z\in[z_1,z_2]}}}\frac{1}{\kappa-1}\left\{ az+\frac{b}{2(\sigma_0^2+\alpha z)}+c  \right\} \,,\label{Power:structure}
  \end{align}
  where
  \begin{equation}\label{abc}
    \left\{
     \begin{aligned}
 a &= \kappa g(t) + \frac{1}{2}{g^2}(t)\sigma _0^2 + \frac{\kappa }{{2\alpha }} \,,\\
 b &= \kappa {\left( {{\mu _0} - \frac{{\sigma _0^2}}{\alpha }} \right)^2} \,,\\
 c& = \kappa {\mu _0}g(t) + \frac{1}{2}\sigma _0^2{g^2}(t) + \frac{{\kappa \sigma _0^2}}{{2{\alpha ^2}}} + \frac{{\kappa \left( {\alpha {\mu _0} - \sigma _0^2} \right)}}{{{\alpha ^2}}}\,.
  \end{aligned}
    \right.
  \end{equation}

%
 For any given set of  the parameters $(\sigma_0,\mu_0,\kappa,\alpha,z_1,z_2,g)$,
 we can easily solve the problem \eqref{Power:structure} with respect to $z\in[z_1,z_2]$, and the optimal investment strategy is correspondingly given as
 the expression \eqref{Strategy:vol}.
 The analytical expression for   $z^*$ is omitted here, since it is  not very expressive, in the sense that
    the solution for  \eqref{Power:structure} does not provide a straightforward intuition for the determinants of the conservative beliefs.
    Obviously, the value of $z^*$ depends on the interval $[z_1,z_2]$ and the shape of \eqref{Power:structure}.
    To derive more intuitional  information on the conservative beliefs and its dependence on  $z^*$,
    we define
    \begin{equation}
      \hat f(z)=\frac{1}{\kappa-1}\left\{ az+\frac{b}{2(\sigma_0^2+\alpha z)}+c  \right\} ,
    \end{equation}
    where $a,b,c$ are given in \eqref{abc}.
    The second order derivative $\hat f^{''}$ of $\hat f$ with respect to $z$ is
    \[\hat f^{''}(z)=\frac{\alpha^2b}{(\kappa-1)(\sigma_0^2+\alpha z)^3}\,.\]
 Since $\kappa\in(0,1),  \sigma_0^2+\alpha z>0$, and $b>0$, we have
    \[\hat f^{''}(z)<0\,\mbox{ for } z\in[z_1,z_2]\,.\]
    That is, $\hat f$ is a concave function on $[z_1,z_2]$.
    Such property relates  $z^*$  to the model parameters and the concavity  of $\hat f$,
     as shown in the Figure \ref{Fig1} with some toy examples.

    These toy examples show the concavity  of $\hat f$ in the setting of $\alpha=0.5$ and $\alpha=-0.5$ with the following common
 parameters
  \[
  \kappa=0.4,\quad \mu_0=0.02,\quad\sigma_0^2=0.1,\quad g(t)\equiv0.1\,.
  \]
  For each $\alpha$, we denote by $\tilde{z}$ the value of $z\in [-0.2,0.2]$ at which $\hat f$ reaches its maximum.
  Then,
     we have three cases of $[z_1,z_2]\subseteq[-0.2,0.2]$  for each $\alpha$, \emph{i.e.},  $z_2<\tilde{z}$, $z_1<\tilde{z}<z_2$, and $\tilde{z}<z_1$.
   Take the case of $\alpha=0.5$ and $ z_2<\tilde{z}$ for example,
   $\hat f$ reaches its maximum at $z^*=z_2$ if $z\in[z_1,z_2]$.
   Correspondingly, we have $\mu^*=\overline{\mu}$ and $\sigma^*=\overline{\sigma}^2$.
   One can easily figure  out $z^*$ in the other cases from Figure \ref{Fig1}.
   We summarize these toy examples in Table \ref{opz22}.
  Generally speaking,
 $z^*$ may take the upper or lower bound of the interval for $z$, or some value lying in the interval.
  When the mean return is positively related to the volatility of the risky asset $(\alpha>0)$, the worst-case scenario of   these two parameters is $(\underline{\mu},\underline{\sigma}^2)$, $(\overline{\mu},\overline{\sigma}^2)$, or some intermediate value depending on some $\tilde{z}\in[z_1,z_2]$.
  When they are  negatively related,
  the conservative belief is  $(\underline{\mu},\overline{\sigma}^2)$, $(\overline{\mu},\underline{\sigma}^2)$, or some intermediate value depending on some $\tilde{z}\in[z_1,z_2]$.

  \begin{table}[h]
  \centering
  \caption{Conservative belief on the mean return and the volatility}\label{opz22}
  \vspace{3mm}
    \setlength{\tabcolsep}{6mm}{
    \begin{tabular}{|c|c|c|c|c|c|c|}
     \hline
      & \multicolumn{3}{c|}{$\alpha>0$} &\multicolumn{3}{c|}{$\alpha<0$}   \\
      \hline
    $z^*$ & $z_1$ &$\tilde{z}\in (z_1,z_2)$& $z_2$ & $z_{1}$ & $\tilde{z}\in (z_1,z_2)$& $z_{2}$ \\
      \hline
      $\mu^*$ & $\underline{\mu}$  & $\mu_0+r-\tilde{z}$&$\overline{\mu}$   & $\underline{\mu}$   &$\mu_0+r-\tilde{z}$&$\overline{\mu}$    \\
      ${\sigma^2}^*$ &  $\underline{\sigma}^2$ & $\sigma_0^2+\alpha \tilde{z}$&$\overline{\sigma}^2$ & $\overline{\sigma}^2$ & $\sigma_0^2+\alpha \tilde{z}$ &$\underline{\sigma}^2$\\
      \hline
   \end{tabular}}
\end{table}

By specifying the interval $[z_1,z_2]$,
we can not only verify the conservative belief on $(\mu,\sigma^2)$ given in Table \ref{opz22} or Figure \ref{Fig1}, but also the relation between trading direction and total risk premium.
Some alternatives for $[z_1,z_2]$  are given in Table \ref{opz2}.
The worst-case scenario  $(\mu^*,\sigma^*)$ is consistent with the implications of Table \ref{opz22} or Figure \ref{Fig1}.
The corresponding investment strategy and total risk premium   listed in the last two columns  show that the investor will take a long position on the risky assets if the total risk premium is position, and vice versa.
This is consistent with our theoretical statements, as given by \eqref{Strategy:vol2}.

\begin{table}[h]
  \centering
  \caption{Utility parameters and the corresponding worst-case mean return and volatility}\label{opz2}
  \vspace{3mm}
   \begin{tabular}{    r r r c c c r c}
     \toprule
         \multicolumn{1}{ c } { $z_1$} & \multicolumn{1}{ c } {$z_2$} & \multicolumn{1}{ c } {$\alpha$} & $z^*$&$u^*$ & $\sigma^*$ &\multicolumn{1}{ c }{$\pi^*$} & Total Risk Premium\\
     \midrule
                  -0.15   &  -0.08 &  0.5 & $z_2$ &$\overline{\mu}$& $\overline{\sigma}^2$&-0.0795 &$-$\\
              -0.08 &0.07&0.5&  -0.0289& $r -$0.0089  & $ \sigma_0^2-$0.0145&-0.0059 &$-$\\
               0.02& 0.12& 0.5 &$z_1$&$\underline{\mu}$ & $\underline{\sigma}^2$&0.7727 &$+$\\
                  -0.15 &   -0.08&   -0.5 & $z_2$ &$\overline{\mu}$& $\underline{\sigma}^2$& -0.5476 &$-$\\
               -0.08  &0.07&-0.5& -0.0311 & $r-$  0.0111& $ \sigma_0^2+$0.0156&0.0066 &$+$\\
                0.02& 0.12& -0.5 &$z_1$& $\underline{\mu}$ & $\overline{\sigma}^2$&0.9074 &$+$\\
     \bottomrule
   \end{tabular}
\end{table}

%
%


\begin{figure}
  \centering
  \subfigure[ ]{
  \begin{minipage}[t]{0.45\linewidth}
  \centering
  \includegraphics[scale=0.5]{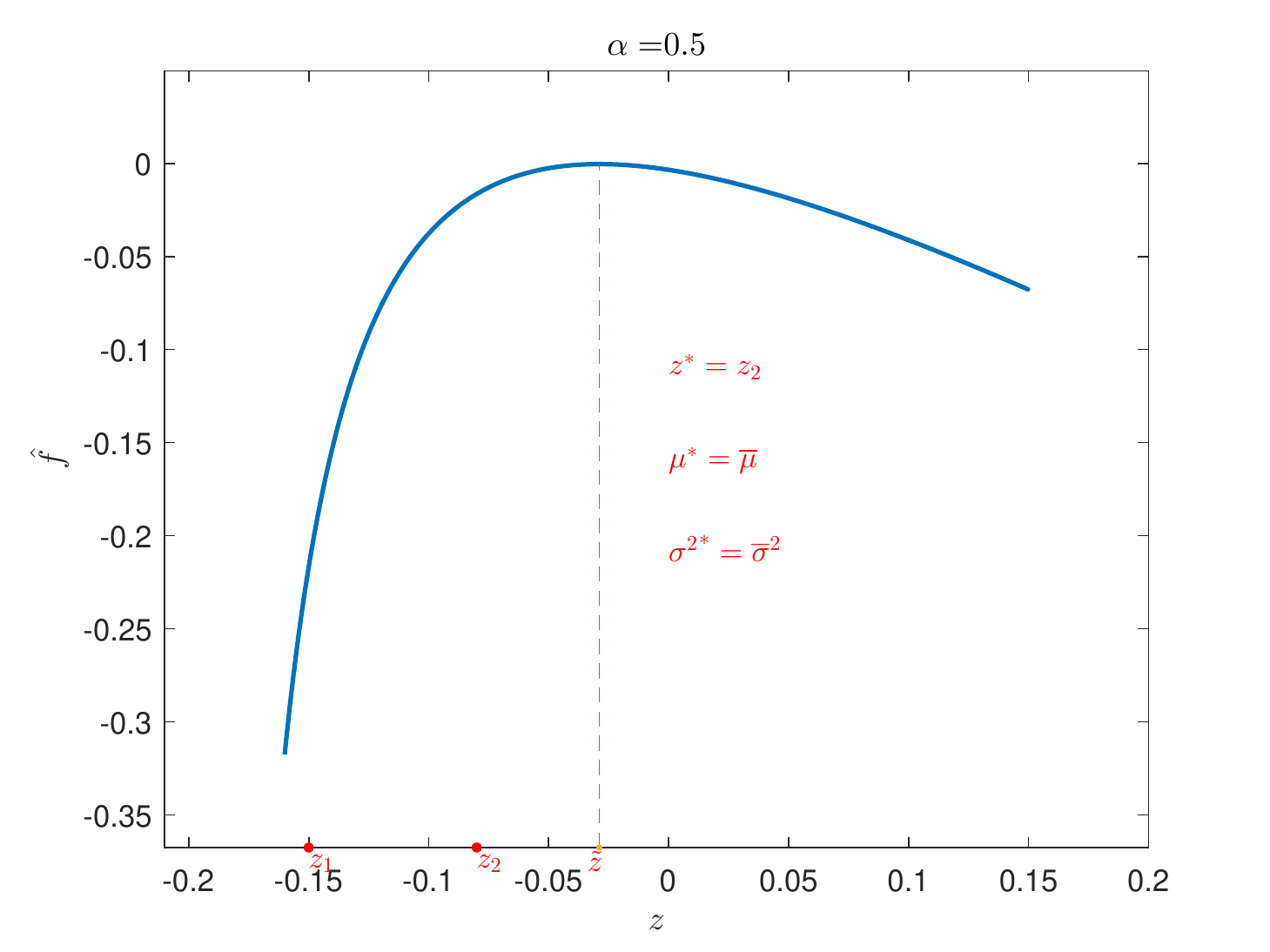}
  \end{minipage}
 }
\subfigure[ ]{
  \begin{minipage}[t]{0.45\linewidth}
  \centering
  \includegraphics[scale=0.5]{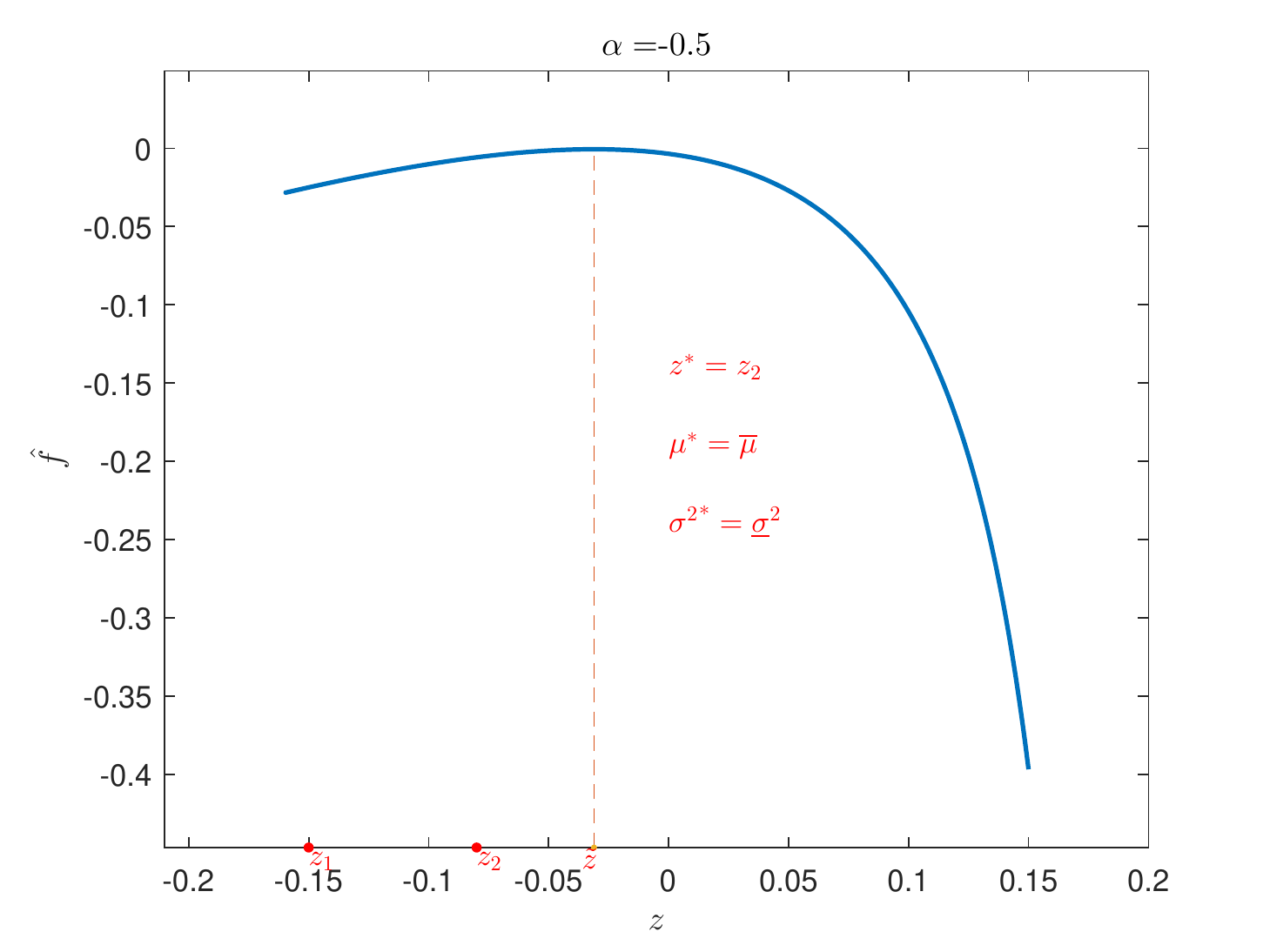}
  \end{minipage}
 }
   \subfigure[ ]{
  \begin{minipage}[t]{0.45\linewidth}
  \centering
  \includegraphics[scale=0.5]{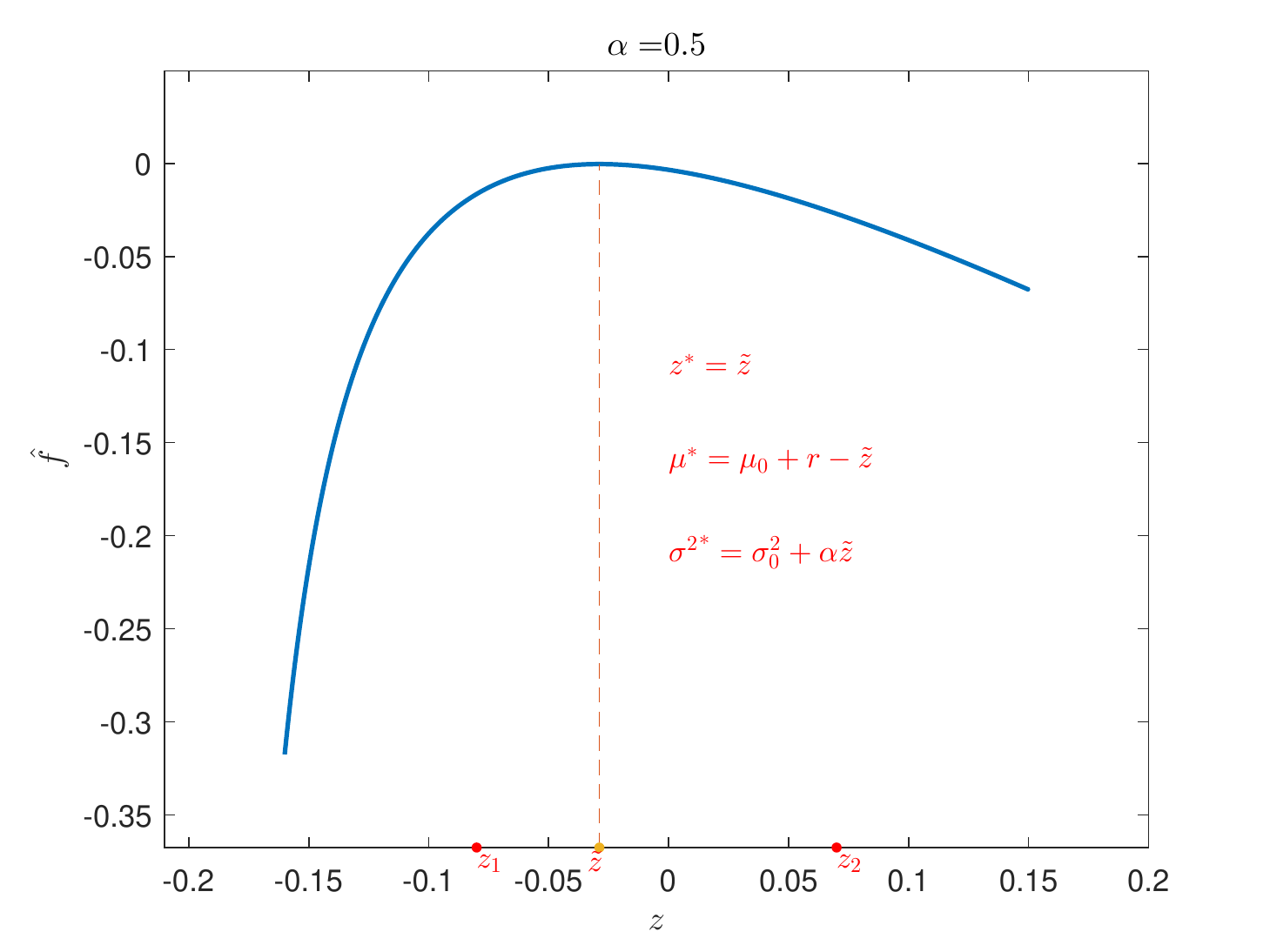}
  \end{minipage}
 }
\subfigure[ ]{
  \begin{minipage}[t]{0.45\linewidth}
  \centering
  \includegraphics[scale=0.5]{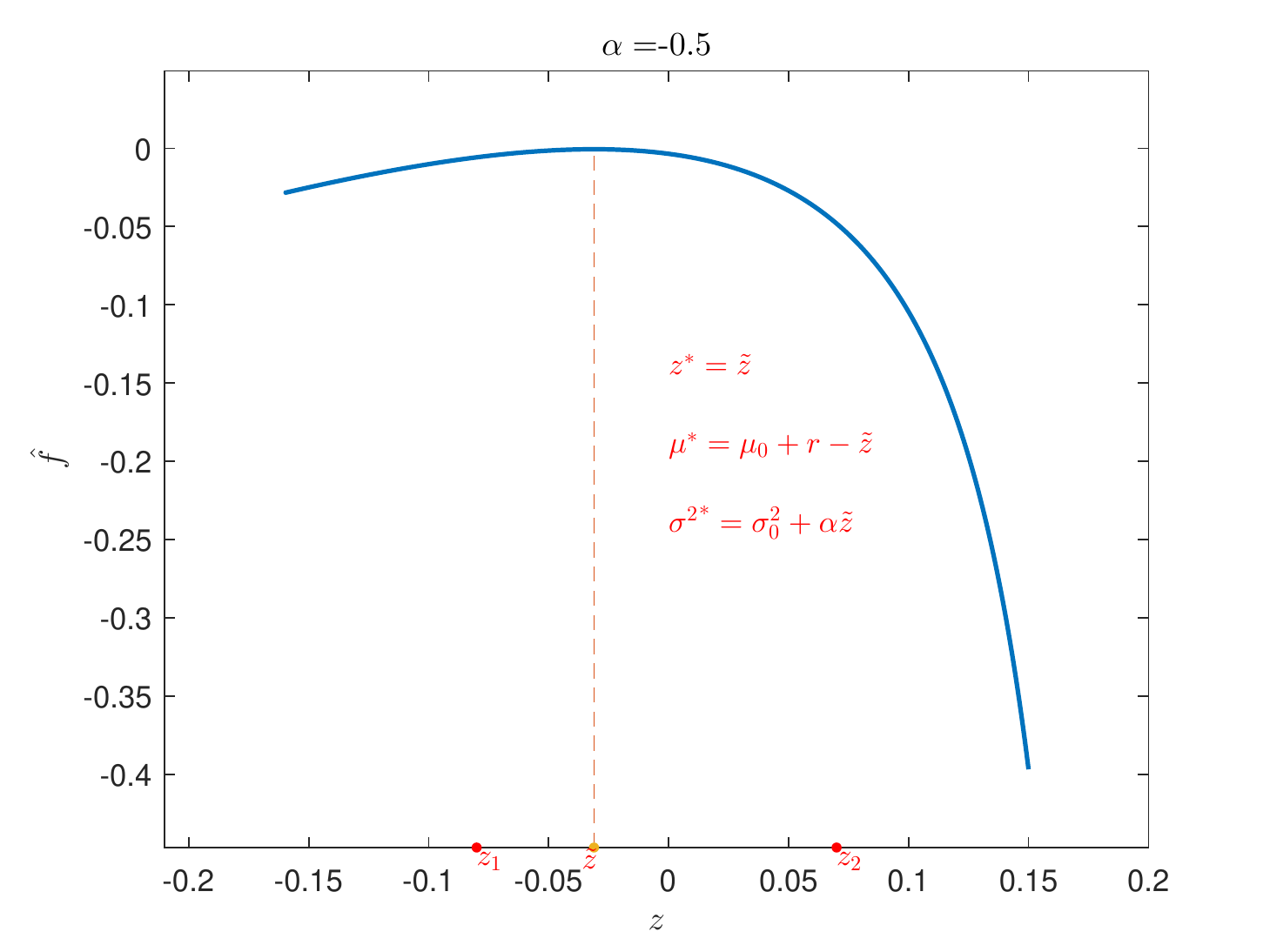}
  \end{minipage}
 }
    \subfigure[ ]{
  \begin{minipage}[t]{0.45\linewidth}
  \centering
  \includegraphics[scale=0.5]{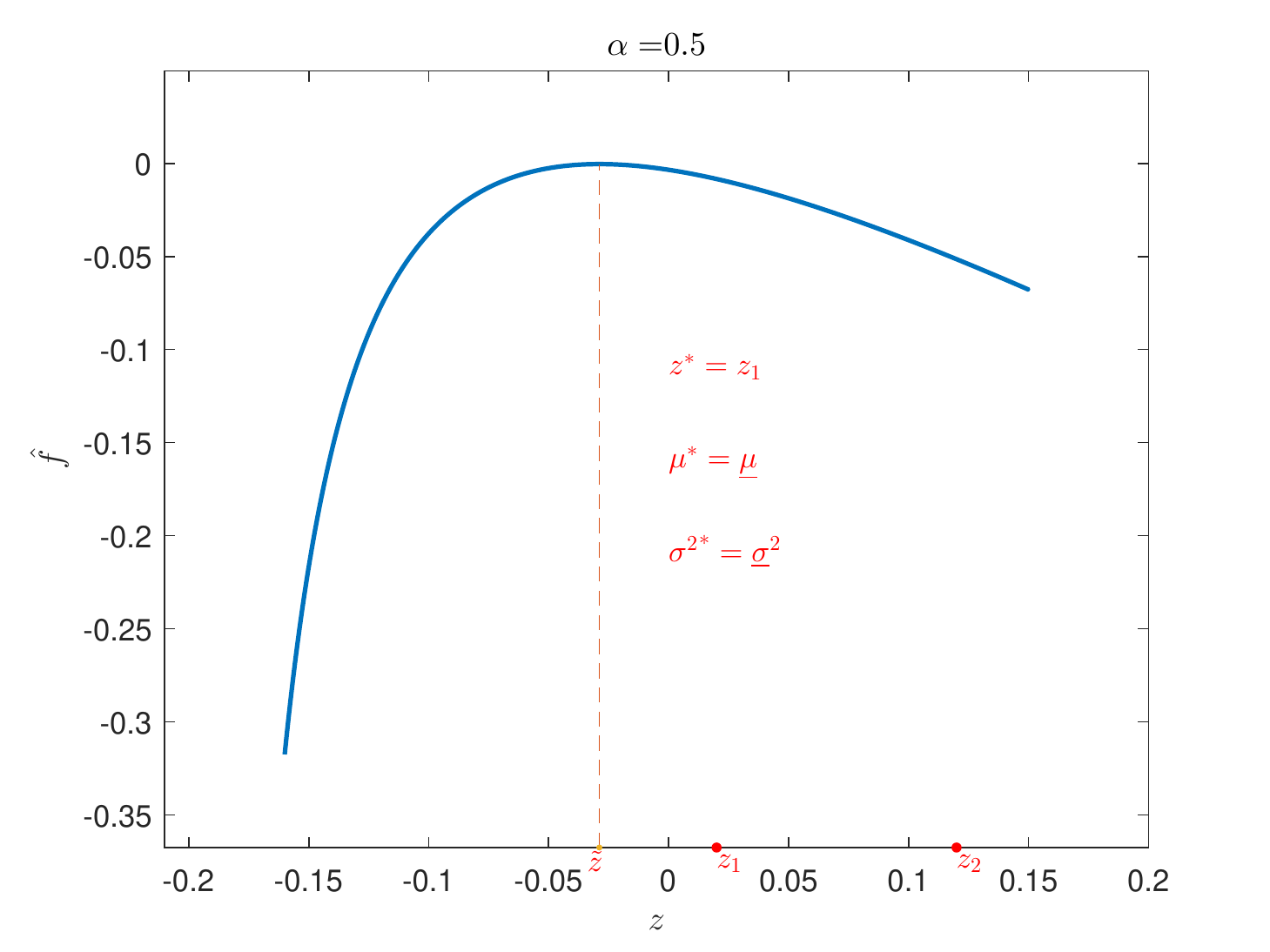}
  \end{minipage}
 }
\subfigure[ ]{
  \begin{minipage}[t]{0.45\linewidth}
  \centering
  \includegraphics[scale=0.5]{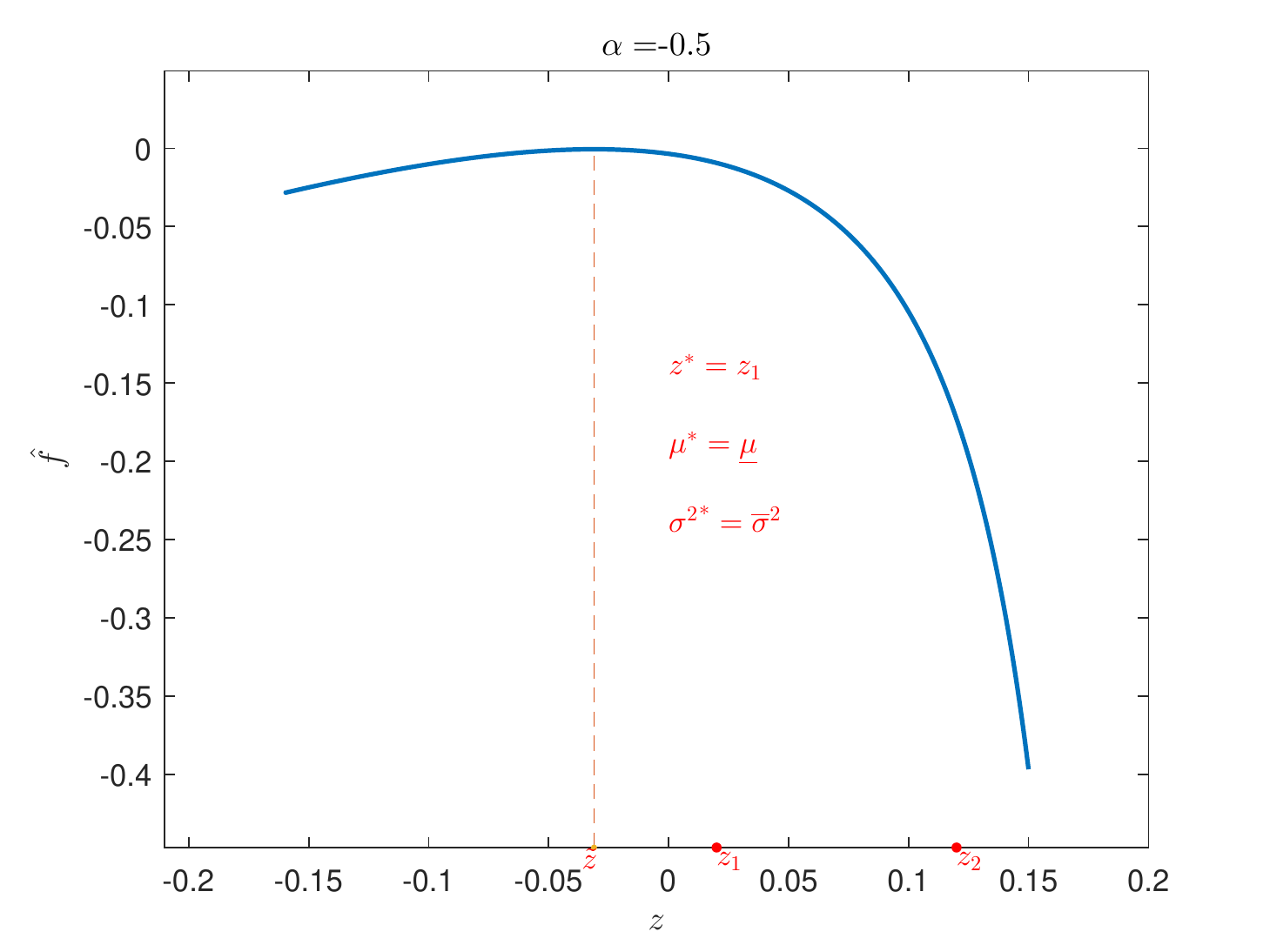}
  \end{minipage}
 }
 \caption{optimal vale $z^*$ in different settings.}\label{Fig1}
\end{figure}

\section{Conclusion}
\label{Conclusion}
The complicated market confronts an investor to ambiguity on the driving force of market randomness.
Such ambiguity may take the form of ambiguity on the mean return rate and volatility of an risky asset.
It may also affect the investor's preference when making investing decisions.
That is, the investor may be ambiguous not only on the characteristics of risky  assets but also  on her future preference.
  We took these two types of ambiguity into account, and
 investigated the horizon-unbiased investment problem.

We proposed the robust forward performance measure by accounting for an investor's ambiguity on the future preference, arising from the ambiguity on the driving force of market randomness.
This robust forward performance measure is then applied to formulate the investment problem.
The solution to such investment problem shows that
the sum of the market risk premium and the utility risk premium jointly determines the optimal trading direction.
If it is positive, the investor will take a long position on the risky asset.
Otherwise, she will take a short position.
This statement holds regardless of the specific form of the forward performance measures.

We then explored the  worst-case scenarios of the mean return and volatility when the initial utility is of the CRRA type.
Specifically, we investigate the  worst-case mean return and volatility in three settings:
ambiguity on mean return $\mu$,
 ambiguity on the volatility $\sigma$, and ambiguity on both mean return and volatility.
In the case of ambiguity on the mean return,
it is the total value of the market risk premium and the utility risk premium that determines an investor's conservative belief;
In the case of ambiguity on the volatility, it is the relative value of these two premiums that affects an investor's conservative belief.
In the case of ambiguity on both the mean return and   volatility,
the conservative belief may not be directly associated with these two premiums.
Note that, in all the three settings, the conservative beliefs may be the some intermediate values within their candidate value intervals, as well as boundaries.

In conclusion, the results provide explanations on the  mechanism of conservative belief selection and robust portfolio choice when an investor propagates her preference in accordance with the market evolution.

 \newpage
\renewcommand{\baselinestretch}{1}
\bibliographystyle{apalike2}

\end{document}